\documentclass{IEEEtaes}

\usepackage{color,array,amsthm}
\usepackage{amsmath,amssymb,amsfonts}
\usepackage{graphicx}
\usepackage{cite}
\usepackage{multirow}
\usepackage{booktabs}
\usepackage{comment}
\usepackage{soul}
\usepackage{stfloats}       

\sethlcolor{yellow}

\allowdisplaybreaks

\newtheorem{theorem}{Theorem}           
\newtheorem{lemma}{Lemma}               
\newtheorem{assum}{Assumption}          
\newtheorem{rem}{Remark}                

\begin{document}

\title{Integrated Flight and Propulsion Control for Fixed-Wing UAVs via Thrust and Disturbance Compensation}

\author{CHONG-YI SUN}
\member{Student Member, IEEE}
\affil{Dalian University of Technology, Dalian, China}

\author{HELING YUAN}
\affil{Dalian University of Technology, Dalian, China}

\author{XU FANG}
\member{Member, IEEE}
\affil{Dalian University of Technology, Dalian, China}

\author{YAN HE}
\affil{Dalian University of Technology, Dalian, China}

\author{XI-MING SUN}
\member{Senior Member, IEEE}
\affil{Dalian University of Technology, Dalian, China}

\authoraddress{Chong-Yi Sun is with the School of Control Science and Engineering, Dalian University of Technology, Dalian 116024, China (e-mail: \href{mailto:dut\_suncy@mail.dlut.edu.cn}{dut\_suncy@mail.dlut.edu.cn}).}

\maketitle

\begin{abstract}
    This paper investigates the position-tracking control problem for fixed-wing unmanned aerial vehicles (UAVs) equipped with a turbojet engine via an integrated flight and propulsion control scheme. To this end, a hierarchical control framework with thrust and disturbance compensation is proposed. In particular, we first propose a perturbed fixed-wing UAV model with turbojet engine dynamics, accounting for both unmodeled dynamics and external disturbances. Second, a versatile extended observer is designed to handle both unmeasurable thrust dynamics and external disturbances. Third, a hierarchical control framework is implemented using three observer-based controllers to guarantee position-tracking performance. With the proposed control strategy, we prove that the closed-loop system asymptotically converges to the desired trajectory. Finally, a comparative simulation is performed to illustrate the proposed control strategy.
\end{abstract}


\section{INTRODUCTION}
\label{sec-1}

The excellent characteristics of fixed-wing unmanned aerial vehicles (UAVs) in high efficiency, long flight duration, and high cruising speed have attracted considerable interest in recent years \cite{Beard2012}. Currently, the applications of fixed-wing UAVs can be found in many fields, including cooperative missions \cite{Wu2025,Zhao2024}, surveillance \cite{Hu2022,Yuan2023}, target enclosing \cite{Peng2024,Peng2025}, and air traffic management \cite{Cerezo2021,Cerezo2022}. In these applications, motion control issues in fixed-wing UAVs are indispensable, including attitude and position-tracking control.

The main challenge in motion control is addressing the strong nonlinearities and aerodynamic disturbances of fixed-wing UAVs. To address these nonlinearities and uncertainties, a typical control strategy can be found in \cite{Castañeda2017}, where a backstepping control framework with an extended state observer is used to construct an attitude and airspeed controller. The main idea in \cite{Castañeda2017} is combined with other approaches like the gain scheduling \cite{Poksawat2018,Shao2025}, the sliding mode control \cite{Babaei2019,Nguyen2023}, the observer-based control \cite{Zhao2024-2,Xu2024,Fu2025}, and the adaptive control \cite{Hu2024}. These control strategies, however, often neglect the complex dynamics of propulsion systems, such as turbojet engines. The hysteresis and instability phenomena in propulsion systems, like rotating stall and surge \cite{Sun2025}, which in turn deteriorate the tracking control performance. Therefore, it remains unclear how these methods perform in the presence of such propulsion dynamics.

To solve the above issues, an integrated flight and propulsion control framework for fixed-wing UAVs is proposed \cite{Garg1993}, in which the propulsion dynamics is addressed to design an integrated controller for fixed-wing UAVs equipped with turbojet engines. Specifically, compared with \cite{Castañeda2017}, a dynamics model of the engine thrust is established with fuel flow as the input, and the thrust subcontroller is incorporated into the integrated controller. For example, the propulsion system is modeled as a second-order system in \cite{Wu2016}, and a thrust controller is integrated in the flight controller. To address the propulsion dynamics, many techniques are introduced, including backstepping \cite{Zhang2019,Sun2020-2}, sliding mode control \cite{Wang2017,Chen2023,Sun2020}, adaptive control \cite{Dou2017,Guo2020,Hu2018}, and observer-based control \cite{Xu2017,Yang2013,An2016}. Despite the merits of these methods, the simplification of second-order systems is applied in propulsion systems, which overlooks the inherent pneumatic and thermodynamic properties of real turbojet engines. Furthermore, these existing methods can only be used to address the 2D longitudinal motion dynamics, while how to resolve the 3D motion dynamics control problem is unknown yet.

Motivated by the above discussions, this paper focuses on 3D position-tracking control for fixed-wing UAVs equipped with turbojet engines. A hierarchical integrated flight and propulsion control framework with thrust and disturbance compensation is proposed. To this end, we first develop a perturbed fixed-wing UAV model that accounts for turbojet engine dynamics, unmodeled dynamics, and external disturbances. Specifically, the turbojet engine is modeled considering pneumatic and thermodynamic properties. Second, with the developed model, we establish a versatile extended observer to estimate the external disturbance and modeling errors, including the unmeasurable thrust. Third, a hierarchical control framework is implemented to establish an integrated flight and propulsion controller for guaranteeing the position-tracking performance. In particular, three observer-based controllers, including the motion state planner, the flight direction controller, and the airspeed controller, are designed with disturbance compensation. With the proposed control framework, the fixed-wing UAV asymptotically converges to the desired position trajectory. Finally, a numeric simulation is implemented to illustrate the proposed control framework. The main contributions of this paper are outlined as follows.
\begin{enumerate}
    \item A perturbed integrated flight and propulsion model is proposed for fixed-wing UAVs. Compared with most existing ones \cite{Castañeda2017,Poksawat2018,Shao2025,Babaei2019,Nguyen2023,Zhao2024-2,Xu2024,Fu2025,Hu2024} in which the propulsion dynamics is neglected, both turbojet engine dynamics and flight kinetics are included, thereby resulting in more difficulties in system analysis and controller design.
    \item A versatile extended observer is designed for both the unmeasurable thrust dynamics and the external disturbances. With the proposed observer, the immeasurability of thrust is addressed.
    \item A hierarchical integrated flight and propulsion control framework is proposed to guarantee the position-tracking performance. Both fixed-wing UAV and turbojet engine dynamics are considered.
\end{enumerate}

The remainder of this paper is organized as follows. In Section~\ref{sec-2}, the models of fixed-wing UAVs and turbojet engines are introduced. The integrated flight and propulsion control framework is proposed in Section~\ref{sec-3}. Section~\ref{sec-4} provides the numeric simulation. Some conclusions are drawn in Section~\ref{sec-5}.

{\bf Notations:} Let $\mathbb{R} := (-\infty, +\infty)$, $\mathbb{R}^+ := [0, +\infty)$ and $\mathbb{R}^n$ denote the $n$-dimensional Euclidean space. Throughout this paper, the argument of a function is sometimes omitted when no confusion occurs, for example, $h(t)$ and $h$ refer to the same function.

\section{FLIGHT AND PROPULSION SYSTEMS}
\label{sec-2}

In this section, we introduce an integrated flight and propulsion model, comprising a fixed-wing UAV and a turbojet engine. The decomposed form of the model is also presented to facilitate the controller design.

\subsection{Fixed-wing UAV Model}
The following equations govern the kinetics of the fixed-wing UAV in the presence of external disturbances \cite{Castañeda2017,Stevens2015}.
\begin{align}
    \dot{x} &= V \cos\gamma \cos\chi + d_x, \notag \\
    \dot{y} &= V \cos\gamma \sin\chi + d_y, \notag \\
    \dot{z} &= -V \sin\gamma + d_z, \notag \\
    \dot{V} &= \frac{T \cos\alpha\cos\beta - D}{m} - g\sin\gamma + d_V, \notag \\
    \dot{\gamma} &= \frac{(T\cos\alpha\sin\beta - Y)\sin\mu + (T\sin\alpha + L)\cos\mu}{mV} \notag \\
    & \quad - \frac{g\cos\gamma}{V} + d_\gamma, \notag \\
    \dot{\chi} &= \frac{(-T\cos\alpha\sin\beta + Y)\cos\mu }{mV\cos\gamma}\!+\!\frac{(T\sin\alpha + L)\sin\mu}{mV\cos\gamma}\!+\!d_\chi, \notag \\
    \dot{\boldsymbol{\omega}} &= \boldsymbol{J}^{-1}\boldsymbol{M} - \boldsymbol{J}^{-1} (\boldsymbol{\omega} \times \boldsymbol{J}\boldsymbol{\omega}) + \boldsymbol{d}_\omega,
    \label{equ-1}
\end{align}
where $x, y, z \in \mathbb{R}$ are the position coordinates of the UAV, $V \in \mathbb{R}^+$ is the airspeed, $\gamma \in \mathbb{R}$ is the flight path angle (FPA), $\chi \in \mathbb{R}$ is the heading angle, $\boldsymbol{\omega} = (p, q, r)^T$ is the angular velocity, in which $p \in \mathbb{R}$ is the roll rate, $q \in \mathbb{R}$ is the pitch rate, and $r \in \mathbb{R}$ is the yaw rate. $d_x, d_y, d_z, d_V, d_\gamma, d_\chi \in \mathbb{R}$ and $\boldsymbol{d}_\omega \in \mathbb{R}^3$ are the unknown but bounded external disturbances. $g \in \mathbb{R}^+$ is the gravitational acceleration, $\boldsymbol{J} \in \mathbb{R}^{3\times 3}$ is the inertia tensor, and $m \in \mathbb{R}^+$ is the mass of the UAV. $\alpha, \beta, \mu \in \mathbb{R}$ are aerodynamic Euler angles, i.e., the angle of attack (AOA), sideslip angle, and bank angle, which are governed by the following equations.

\begin{align}
    \label{equ-add1}
    \dot{\alpha} &= q - \frac{\cos\alpha\sin\beta}{\cos\beta} \cdot p - \frac{\sin\alpha\sin\beta}{\cos\beta} \cdot r - \frac{\cos\mu}{\cos\beta} \cdot \dot{\gamma} \notag \\
    & \quad - \frac{\cos\gamma\sin\mu}{\cos\beta} \cdot \dot{\chi} + d_\alpha, \notag \\
    \dot{\beta} &= \sin\alpha \cdot p - \cos\alpha \cdot r - \sin\mu \cdot \dot{\gamma} + \cos\gamma\cos\mu \cdot \dot{\chi} + d_\beta, \notag \\
    \dot{\mu} &= \cos\alpha\cos\beta \cdot p + \sin\beta \cdot q + \sin\alpha\cos\beta \cdot r \notag \\
    & \quad + \sin\beta\cdot\dot{\alpha} + \sin\gamma\cdot\dot{\chi} + d_\mu,
\end{align}
where $d_\alpha, d_\beta, d_\mu \in \mathbb{R}$ are the unknown but bounded external disturbances. $L, Y, D \in \mathbb{R}$ in \eqref{equ-1} are the aerodynamic forces, i.e., the lift, side force, and drag. The aerodynamic forces are defined as follows.
\begin{equation}
    \label{equ-2}
    L = \bar{q} S C_D, \quad Y = \bar{q} S C_Y, \quad D = \bar{q} S C_D,
\end{equation}
where $S \in \mathbb{R}^+$ is the reference area. The aerodynamic force coefficients $C_L, C_Y, C_D \in \mathbb{R}$ and the dynamic pressure $\bar{q} \in \mathbb{R}$ are introduced in \cite{Stevens2015}, i.e.
\begin{align}
    \bar{q} &= 0.5\rho V^2, \quad \rho = \rho_0 e^{-(h - h_0) / h_e}, \notag \\
    C_L &= C_{L0} + C_{L\alpha}\alpha + \frac{\bar{c}}{2V}(C_{L\dot{\alpha}} \dot{\alpha} + C_{Lq} q) + C_{L\delta e} \delta_e, \notag \\
    C_D &= C_{D0} + \frac{(C_L - C_{L0})^2}{\pi e AR} + C_{D\delta e} \delta_e + C_{D\delta a} \delta_a + C_{D\delta r} \delta_r, \notag \\
    C_Y &= C_{Y\beta} \beta + \frac{b}{2V} (C_{Yp} p\!+\!C_{Yr} r)\!+\!C_{Y\delta a} \delta_a\!+\!C_{Y\delta r} \delta_r,
    \label{equ-3}
\end{align}
where $\delta_e \in \mathbb{R}$ is the elevator, $\delta_a \in \mathbb{R}$ is the aileron, and $\delta_r \in \mathbb{R}$ is the rudder.

The inputs of \eqref{equ-1} are the thrust $T \in \mathbb{R}^+$ and the aerodynamic moment $\boldsymbol{M} = (M_x, M_y, M_z)^T$. The thrust $T$ is generated by the propulsion system, which will be introduced in the next section. The following definitions of $M_x, M_y$, and $M_z$ are adopted \cite{Stevens2015}.
\begin{equation}
    \label{equ-4}
    M_x = \bar{q} S b C_l, \quad M_y = \bar{q} S \bar{c} C_m, \quad M_z = \bar{q} S b C_n,
\end{equation}
where
\begin{align}
    C_l &= C_{l\beta} \beta + \frac{b}{2V}(C_{lp} p + C_{lr} r) + C_{l\delta a} \delta_a + C_{l\delta r} \delta_r, \notag \\
    C_m &= C_{m0} + C_{m\alpha} \alpha + \frac{\bar{c}}{2V} (C_{mq} q + C_{m\dot{\alpha}} \dot{\alpha}) + C_{m\delta e} \delta_e, \notag \\
    C_n &= C_{n\beta} \beta + \frac{b}{2V} (C_{np} p + C_{nr} r) + C_{n\delta a} \delta_a + C_{n\delta r} \delta_r,
    \label{equ-5}
\end{align}
are the aerodynamic moment coefficients. In \eqref{equ-3} and \eqref{equ-5}, $\delta_e$, $\delta_a$, and $\delta_r$ are the actuators to handle the UAV states, which are called the flap deflection angles.

\subsection{Turbojet Engine Model}

The turbojet engine model consists of component equations for the inlet, compressor, burner, turbine, nozzle, and rotor, all of which are combined using continuum equations \cite{Wang2022}.

{\it \textbf{1) Inlet:}} The environmental temperature $T_{s0} \in \mathbb{R}^+$ and pressure $P_{s0} \in \mathbb{R}^+$ are
\begin{align}
    \label{equ-6}
    T_{s0} &= \begin{cases}
        288.15 - 6.5H, & H \leq 11, \\
        216.5, &H > 11
    \end{cases} \notag \\
    P_{s0} &= \begin{cases}
        1.01325 \Big(1 - \frac{H}{44.308}\Big)^{5.2553}, & H \leq 11 \\
        0.22615 e^{(11 - H) / 6.338}, & H > 11
    \end{cases}
\end{align}
where $H = -z \times10^{-3}$ is the altitude.
The Mach number ${\rm Ma} \in \mathbb{R}^+$ is defined as
\begin{equation}
    \label{equ-7}
    {\rm Ma} := V / \sqrt{\kappa R T_{s0}},
\end{equation}
where $\kappa \in \mathbb{R}^+$ is the heat capacity ratio, and $R \in \mathbb{R}^+$ is the gas constant. The inlet total temperature is
\begin{equation}
    \label{equ-8}
    T_{t1} = T_{s0} \left( 1 + \frac{\kappa - 1}{2} {\rm Ma}^2 \right),
\end{equation}
and the inlet total pressure is
\begin{equation}
    \label{equ-9}
    P_{t1} = P_{s0} \left( 1 + \frac{\kappa - 1}{2} {\rm Ma}^2 \right)^{\frac{\kappa}{\kappa - 1}}.
\end{equation}
The inlet pressure recovery coefficient $\sigma_{\rm in} \in \mathbb{R}^+$ is defined as
\begin{equation}
    \label{equ-10}
    \sigma_{\rm in} := \begin{cases}
        1, & {\rm Ma} \leq 1.0, \\
        1 - 0.075({\rm Ma} - 1)^{1.35}, & {\rm Ma} > 1.0.
    \end{cases}
\end{equation}
The exit total temperature $T_{t2} \in \mathbb{R}^+$ and the exit total pressure $P_{t2} \in \mathbb{R}^+$ of the inlet are
\begin{equation}
    \label{equ-11}
    T_{t2} = T_{t1}, \quad P_{t2} = P_{t1} \sigma_{\rm in}.
\end{equation}

{\it \textbf{2) Compressor:}} The corrected speed of the compressor $n_{c,{\rm cor}} \in \mathbb{R}^+$ is defined as
\begin{equation}
    \label{equ-12}
    n_{c,{\rm cor}} := n \sqrt{T_{t2d}} / \big(n_d \sqrt{T_{t2}} \big),
\end{equation}
where $n_{c,{\rm cor}} \in \mathbb{R}^+$ is the corrected speed, $n \in \mathbb{R}^+$ is the speed of the rotor, $n_d$ is a constant representing the designed point of the rotating speed, and $T_{t2d}$ is a constant which represents the designed point of the compressor total temperature.
Define the pressure ratio coefficient $z_c \in [0, 1]$ as
\begin{equation}
    \label{equ-13}
    z_c := \frac{\pi_c - \pi_{c,{\rm min}}}{\pi_{c,{\rm max}} - \pi_{c,{\rm min}}},
\end{equation}
where $\pi_{c,{\rm min}}$ and $\pi_{c,{\rm max}}$ are constants which represent the minimum and maximum of the pressure ratio, respectively. The pressure ratio $\pi_c \in \mathbb{R}^+$, corrected flow $W_{a2,cor} \in \mathbb{R}^+$, and the efficient $\eta_c \in \mathbb{R}^+$ are
\begin{align}
    \label{equ-14}
    \pi_c &= f_{c,1} (n_{c,{\rm cor}}, z_c), \notag \\
    W_{a2,cor} &= f_{c,2} (n_{c,{\rm cor}}, z_c), \notag \\
    \eta_c &= f_{c,3} (n_{c,{\rm cor}}, z_c),
\end{align}
where $f_{c,1}, f_{c,2}, f_{c,3}$ are compressor characteristic functions. We further obtain the compressor states as follows.
\begin{align}
    \label{equ-15}
    P_{t3} &= P_{t2} \pi_c, \notag \\
    W_{a3} &= W_{a2} = W_{a2,cor} \frac{P_{t2}}{1.01325\sqrt{288.15/T_{t2}}}, \notag \\
    T_{t3} &= T_{t2} \bigg(\frac{\pi_c^{(\kappa - 1)/\kappa} - 1}{\eta_c} + 1 \bigg), \notag \\
    L_c &= W_{a2} C_p (T_{t3} - T_{t2}),
\end{align}
where $P_{t3} \in \mathbb{R}^+$ is the total pressure of the compressor outlet, $W_{a2} \in \mathbb{R}^+$ is the compressor inlet flow, $W_{a3} \in \mathbb{R}^+$ is the outlet flow, $T_{t3} \in \mathbb{R}^+$ is the total temperature of the compressor outlet, $L_c \in \mathbb{R}^+$ is the power of the compressor, and $C_p$ is a constant which represents the specific heat.

{\it \textbf{3) Burner:}} The gas-oil ratio $f_g \in \mathbb{R}^+$ is defined as
\begin{equation}
    \label{equ-16}
    f_g := W_f / W_{a3},
\end{equation}
where $W_f \in \mathbb{R}^+$ is the fuel flow. The total temperature $T_{t4} \in \mathbb{R}^+$ and total pressure $P_{t4} \in \mathbb{R}^+$ of the burner outlet are
\begin{align}
    \label{equ-17}
    T_{t4} &= T_{t3} + f H_u \eta_b / C_p, \notag \\
    P_{t4} &= \sigma_b P_{t3},
\end{align}
where $H_u \in \mathbb{R}^+$ is the fuel oil low calorific value, $\eta_b \in \mathbb{R}^+$ is the combustion efficiency, and $\sigma_b \in \mathbb{R}^+$ is the burner pressure recovery coefficient.

{\it \textbf{4) Turbine:}} The corrected speed of the turbine $n_{t,{\rm cor}} \in \mathbb{R}^+$ is defined as
\begin{equation}
    \label{equ-18}
    n_{t,{\rm cor}} := n \sqrt{T_{t4d}} / \big(n_d \sqrt{T_{t4}}\big),
\end{equation}
where $T_{t4d} \in \mathbb{R}^+$ is the designed point of the turbine inlet total temperature. The flow ratio coefficient $w_t \in [0, 1]$ is defined as
\begin{equation}
    \label{equ-19}
    w_t := \frac{W_{g4,{\rm cor}} - W_{g4{\rm cor,min}}}{W_{g4{\rm cor,max}} - W_{g4{\rm cor,min}}},
\end{equation}
where $W_{g4{\rm cor,min}}$ and $W_{g4{\rm cor,max}}$ are constants which represent the minimum and maximum of the corrected flow, respectively. The corrected flow $W_{g4,{\rm cor}}$, the corrected power $L_t / T_{t4}$, and the efficient $\eta_t$ are
\begin{align}
    \label{equ-20}
    W_{g4,{\rm cor}} &= f_{t,1} (n_{t,{\rm cor}}, w_t), \notag \\
    L_t / T_{t4} &= f_{t,2} (n_{t,{\rm cor}}, z_c), \notag \\
    \eta_t &= f_{t,3} (n_{t,{\rm cor}}, z_c),
\end{align}
where $f_{t,1}, f_{t,2}, f_{t,3}$ are turbine characteristic functions. We further obtain the turbine states as follows.
\begin{align}
    \label{equ-21}
    L_t &= T_{t4}(L_t / T_{t4}), \notag \\
    T_{t5} &= T_{t4} - L_t / C_p, \notag \\
    \pi_t &= (1 - (1 - T_{t5} / T_{t4}) / \eta_t)^{-\kappa/(\kappa - 1)}, \notag \\
    P_{t5} &= P_{t4} / \pi_t, \notag \\
    W_{g4} &= W_{g4,cor} \frac{P_{t4}}{1.01325\sqrt{288.15 / T_{t4}}},
\end{align}
where $L_t \in \mathbb{R}^+$ is the power of the turbine, $T_{t5} \in \mathbb{R}^+$ is the total temperature of the turbine outlet, $\pi_t \in \mathbb{R}^+$ is the pressure ratio, $P_{t5} \in \mathbb{R}^+$ is the total pressure of the turbine outlet, $W_{g4} \in \mathbb{R}^+$ is the turbine outlet flow. 

{\it \textbf{5) Nozzle:}} The nozzle inlet total temperature $T_{t7} \in \mathbb{R}^+$ and total pressure  $P_{t7} \in \mathbb{R}^+$ are
\begin{equation}
    \label{equ-22}
    T_{t7} = T_{t5}, \quad P_{t7} = P_{t5}.
\end{equation}
The critical pressure ratio of nozzle $\pi_{n,cr} \in \mathbb{R}^+$ is defined as
\begin{equation}
    \label{equ-23}
    \pi_{n,cr} = \bigg(\frac{\kappa + 1}{2}\bigg)^{\kappa/(\kappa - 1)}.
\end{equation}
If $P_{t7} / P_{s0} < \pi_{n,cr}$, we have
\begin{align}
    \label{equ-24}
    P_{s8} &= P_{s0}, \notag \\
    {\rm Ma}_8 &= \sqrt{\frac{2\Big((P_{t7} / P_{s0})^{\frac{\kappa - 1}{\kappa}} - 1\Big)}{\kappa - 1}},
\end{align}
where $P_{s8} \in \mathbb{R}^+$ is the nozzle static pressure, ${\rm Ma}_8 \in \mathbb{R}^+$ is the Mach number of the nozzle flow. Otherwise, if $P_{t7} / P_{s0} \geq \pi_{n,cr}$, we have
\begin{equation}
    \label{equ-25}
    P_{s8} = P_{t7} / \pi_{n, cr}, \quad {\rm Ma}_8 = 1.
\end{equation}
We can further obtain the static temperature of the nozzle $T_{s8} \in \mathbb{R}^+$, the nozzle flow speed $v_8 \in \mathbb{R}^+$, and the nozzle flow $W_{g8} \in \mathbb{R}^+$ as follows.
\begin{align}
    \label{equ-26}
    T_{s8} &= T_{t7} \bigg(1 + \frac{\kappa - 1}{2} {\rm Ma}_8^2 \bigg), \notag \\
    v_8 &= {\rm Ma}_8 \sqrt{\kappa R T_{s8}}, \notag \\
    W_{g8} &= \frac{K_m P_{t7} A_8 Q \sqrt{288.15}}{1.01325\sqrt{T_{t7}}},
\end{align}
where $A_8 \in \mathbb{R}^+$ is the area of the nozzle, and
\begin{align}
    \label{equ-27}
    Q &= {\rm Ma}_8 \bigg(\frac{2}{\kappa+1} \bigg(1 + \frac{\kappa - 1}{2} \bigg) {\rm Ma}_8^2 \bigg)^{-\frac{\kappa + 1}{2(\kappa - 1)}}, \notag \\
    K_m &= \sqrt{\frac{\kappa}{R}\bigg(\frac{2}{\kappa + 1}\bigg)^{\frac{\kappa + 1}{\kappa - 1}}}.
\end{align}
The thrust generated by the turbojet engine is
\begin{equation}
    \label{equ-28}
    T = W_{g8}(v_8 - v_0) + A_8 (P_{s8} - P_{s0}).
\end{equation}

{\it \textbf{6) Rotor:}} The rotor speed $n$ is governed by the following equation.
\begin{equation}
    \label{equ-29}
    \dot{n} = \frac{\eta_m L_t - L_c}{n J_r (\pi/30)^2},
\end{equation}
where $\eta_m \in \mathbb{R}^+$ is the turbine mechanical efficiency, and $J_r \in \mathbb{R}^+$ is the inertia of the rotor.

{\it \textbf{7) Continuum Equations:}} All the above components are combined by the following continuum equations.
\begin{equation}
    \label{equ-30}
    W_{g4} - W_{a3} - W_f = 0, \quad W_{g8} - W_{g4} = 0.
\end{equation}

\subsection{Model Decomposition}
In this subsection, the complex models of the fixed-wing UAV and the turbojet engine are decomposed and linearized to yield several subsystems for controller design. In particular, the model is divided into three subsystems, including position, airspeed, and flight direction.

{\it \textbf{1) Position Subsystem:}} Denote the position of the vehicle as $\boldsymbol{p} = (x, y, z)^T$ and the velocity as $\boldsymbol{v} = (V, \gamma,\chi)^T$. The position subsystem is governed by the following equation.
\begin{equation}
    \label{equ-31}
    \dot{\boldsymbol{p}} = \boldsymbol{v}_e + \boldsymbol{d}_p = f_p(\boldsymbol{v}) + \boldsymbol{d}_p,
\end{equation}
where $\boldsymbol{v}_e = (\dot{x}, \dot{y}, \dot{z})^T$ is the velocity in the earth frame, $\boldsymbol{d}_p = (d_x, d_y, d_z)^T$ is the external disturbance, and
\begin{equation}
    \label{equ-32}
    f_p(\boldsymbol{v}) = \begin{bmatrix}
        V \cos\gamma \cos\chi \\
        V \cos\gamma \sin\chi \\
        -V \sin\gamma
    \end{bmatrix}.
\end{equation}

{\it \textbf{2) Airspeed Subsystem:}} The motion dynamics of the airspeed $V$ is formulated by the following equations.
\begin{align}
    \label{equ-33}
    \dot{V} &= f_v + g_v \Delta T + d_V, \notag \\
    \Delta \dot{n} &= a_n \Delta n + b_n \Delta W_f + d_{\Delta n}, \notag \\
    \Delta T &= c_n \Delta n + d_n \Delta W_f + d_{\Delta T},
\end{align}
where $\Delta T = (T - T_0) / T_0$, $\Delta n = (n - n_0) / n_0$, and $\Delta W_f = (W_f - W_{f0}) / W_{f0}$ with $T_0, n_0, W_{f0} \in \mathbb{R}^+$ as the equilibrium point of the turbojet engine. In \eqref{equ-33}, $d_{\Delta n}, d_{\Delta T} \in \mathbb{R}$ are the linearization errors, $a_n, b_n, c_n, d_n \in \mathbb{R}$ are known coefficients, and
\begin{align}
    \label{equ-34}
    f_v &= - D / m - g \sin\gamma + T_0\cos\alpha \cos\beta / m, \notag \\
    g_v &= T_0 \cos\alpha \cos\beta / m.
\end{align}

{\it \textbf{3) Flight Direction Subsystem:}} Denote the flight direction as $\boldsymbol{\psi} = (\gamma, \chi)^T$. The motion dynamics of the flight direction is formulated as follows.
\begin{align}
    \label{equ-35}
    \dot{\boldsymbol{\psi}} &= f_\psi + g_\psi \boldsymbol{F}_a + \boldsymbol{d}_\psi, \notag \\
    \dot{\boldsymbol{\Theta}} &= f_\theta + g_\theta \boldsymbol{\omega} + \boldsymbol{d}_\Theta, \notag \\
    \dot{\boldsymbol{\omega}} &= f_\omega + g_\omega \boldsymbol{M} + \boldsymbol{d}_\omega,
\end{align}
where $\boldsymbol{\Theta} = (\alpha, \beta, \mu)^T$ is the aerodynamic Euler angle, $\boldsymbol{d}_\psi = (d_\gamma, d_\chi)^T$ and $\boldsymbol{d}_\Theta = (d_\alpha, d_\beta, d_\mu)^T$ are external disturbances. In \eqref{equ-35}, $\boldsymbol{F}_a = (L, Y)^T$ is the aerodynamic driving force, and $\boldsymbol{M} = (M_x, M_y, M_z)^T$ is the aerodynamic moment with the following forms.
\begin{align}
    \label{equ-36}
    \boldsymbol{F}_a &= f_a(\boldsymbol{\theta}) = f_F + g_F \boldsymbol{\theta}, \notag \\
    \boldsymbol{M} &= f_m(\boldsymbol{\delta}) = f_M + g_M \boldsymbol{\delta},
\end{align}
where $\boldsymbol{\delta} = (\delta_e, \delta_a, \delta_r)^T$ is the wing flap, which serves as one of the system inputs. In \eqref{equ-36}, $\boldsymbol{\theta} = (\alpha, \beta)^T$ is the aerodynamic angle, and
\begin{align}
    \label{equ-37}
    f_F &= \bar{q}S \begin{bmatrix}
        C_{L0} + (C_{L\dot{\alpha}} \dot{\alpha} + C_{Lq}q)\bar{c} / 2V + C_{L\delta e} \delta_e \\
        (C_{Yp} p + C_{Yr} r)b/2V + C_{Y\delta a}\delta_a + C_{Y\delta r}\delta_r
    \end{bmatrix}, \notag \\
    g_F &= \bar{q}S \begin{bmatrix}
        C_{L\alpha} & 0 \\
        0 & C_{Y\beta}
    \end{bmatrix}, \notag \\
    f_M &= \bar{q} S \begin{bmatrix}
        C_{l\beta} \beta + (C_{lp} p+ C_{lr} r)b / 2V \\
        C_{m0} + C_{m\alpha} \alpha + (C_{mq} q + C_{m \dot{\alpha}} \dot{\alpha})\bar{c} / 2V \\
        C_{n\beta} + (C_{np} p + C_{nr} r)b / 2V
    \end{bmatrix}, \notag \\
    g_M &= \bar{q}S \begin{bmatrix}
        0 & b C_{l\delta a} & b C_{l \delta r} \\
        \bar{c} C_{m \delta e} & 0 & 0 \\
        0 & b C_{n \delta a} & b C_{n \delta r}
    \end{bmatrix}.
\end{align}
The other terms in \eqref{equ-35} are in the following forms.
\begin{align}
    \label{equ-38}
    f_\psi &= \begin{bmatrix}
        \frac{T(\cos\alpha\sin\beta\sin\mu + \sin\alpha\cos\mu)}{mV} - \frac{g\cos\gamma}{V} \\
        \frac{T(\cos\alpha\sin\beta\cos\mu + \sin\alpha\sin\mu)}{mV\cos\gamma}
    \end{bmatrix}, \notag \\
    g_\psi &= \begin{bmatrix}
        \frac{\cos\mu}{mV} & \frac{\sin\mu}{mV} \\
        \frac{\sin\mu}{mV\cos\gamma} & \frac{\cos\mu}{mV\cos\gamma}
    \end{bmatrix}, \notag \\
    f_\theta &= \begin{bmatrix}
        \dot{\gamma} \cos\mu / \cos\beta + \dot{\chi} \cos\gamma \sin\mu / \cos\beta \\
        -\dot{\gamma} \sin\mu + \dot{\chi} \cos\gamma \cos\mu \\
        \dot{\alpha} \sin\beta + \dot{\chi} \sin\gamma
    \end{bmatrix}, \notag \\
    g_\theta &= \begin{bmatrix}
        \cos\alpha\cos\beta / \sin\beta & 1 & \sin\alpha\sin\beta / \cos\beta \\
        \sin\alpha & 0 & -\cos\alpha \\
        \cos\alpha \cos\beta & \sin\beta & \sin\alpha \cos\beta
    \end{bmatrix}, \notag \\
    f_\omega &= -\boldsymbol{J}^{-1}(\boldsymbol{\omega} \times \boldsymbol{J\omega}), \quad g_\omega = \boldsymbol{J}^{-1}.
\end{align}

The following lemma shows the stability of the cascade system, which is used in our stability analysis.
\begin{lemma}
    \label{lem-1}
    \rm
    \cite{Isidori1995} Consider a cascade system
    \begin{equation}
        \label{equ-38-1}
        \dot{\xi} = g(\xi, \epsilon), \quad \dot{\epsilon} = h(\epsilon).
    \end{equation}
    If $\dot{\xi} = g(\xi, 0)$ is asymptotically stable at $\xi = 0$ and $\dot{\epsilon} = h(\epsilon)$ is asymptotically stable at $\epsilon = 0$, then \eqref{equ-38-1} is asymptotically stable at $(\xi, \epsilon) = (0, 0)$.
\end{lemma}

\subsection{Control Objective}

In this paper, we consider an integrated model of the fixed-wing UAV equipped with a turbojet engine, as presented in \eqref{equ-31}, \eqref{equ-33}, and \eqref{equ-35}. Our goal is to design a tracking controller to guarantee the convergence of the position $\boldsymbol{p}(t)$ to the desired trajectory $\boldsymbol{p}_d(t)$ and the wind-axis bank angle $\mu(t)$ to the desired bank angle $\mu_d(t)$, that is,
\begin{align}
    \label{equ-39}
    \lim_{t \to \infty} \big( \boldsymbol{p}(t) - \boldsymbol{p}_d(t) \big) &= 0, \notag \\
    \lim_{t \to \infty} \big(\mu(t) - \mu_d(t) \big) &= 0.
\end{align}
In \eqref{equ-39}, the desired trajectory $\boldsymbol{p}_d$ and the bank angle $\mu_d$ are designed to be differentiable.

\begin{assum}
    \label{assum-1}
    \rm
    In \eqref{equ-31}, \eqref{equ-33}, and \eqref{equ-35}, the following holds.
    \begin{enumerate}
        \item The derivative of the disturbance exists and is bounded and vanishing. That is, $\dot{d} = h(t)$, where $h(t)$ is an unknown but bounded function and $\lim_{t \to \infty} h(t) = 0$.
        \item The system coefficients $a_n, b_n, c_n, d_n \neq 0$.
    \end{enumerate}
\end{assum}

How to design an appropriate tracking control strategy to achieve this control goal will be addressed in the following sections.

\section{HIERARCHICAL CONTROL STRATEGY}
\label{sec-3}

In this section, a tracking control strategy within the hierarchical framework with disturbance compensation is proposed for driving the fixed-wing UAV, as outlined in Fig.~\ref{fig-1}. We first present the extended disturbance observer and the tracking controller design in detail, and then conduct the stability analysis of the closed-loop system.

\begin{figure*}
    \centering
    \includegraphics[width=18cm]{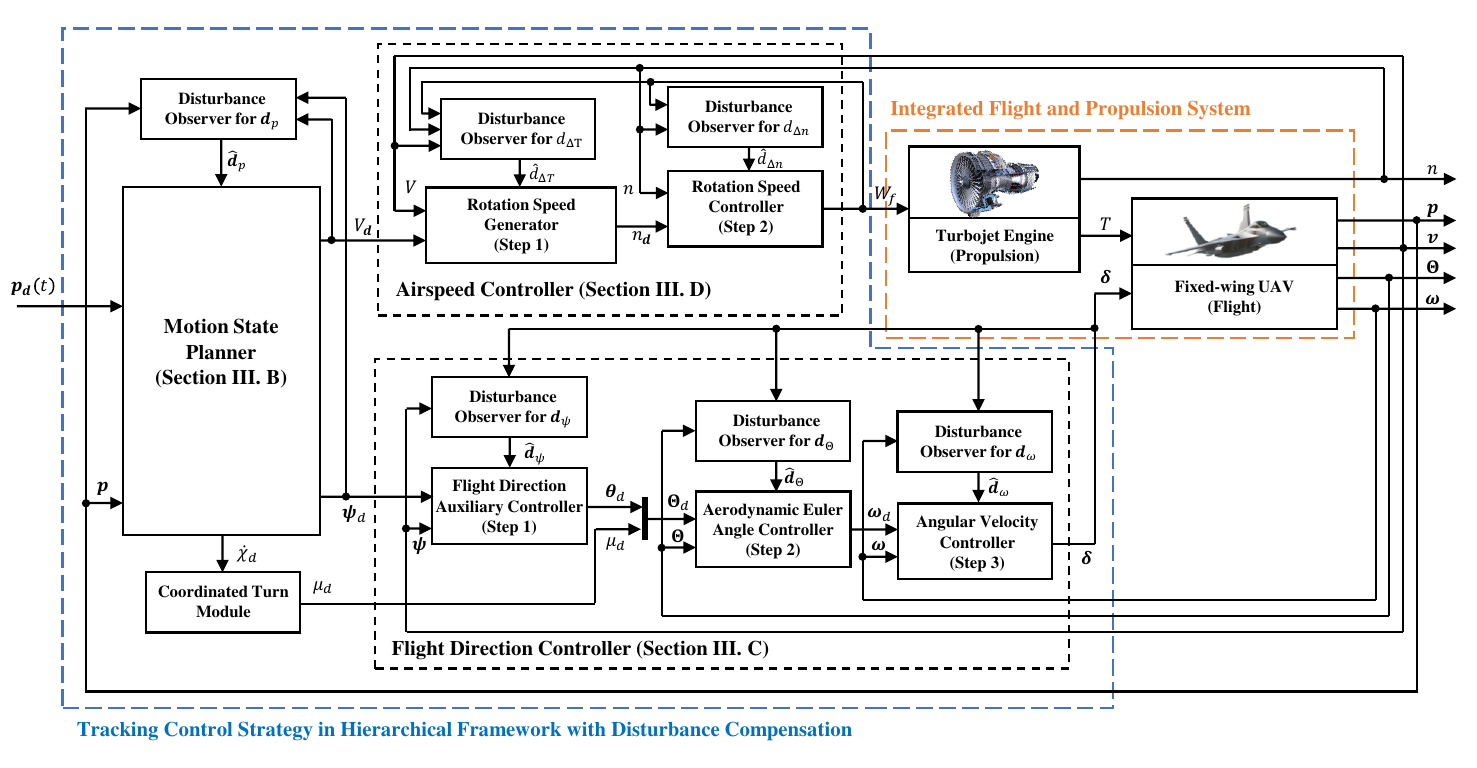}
    \caption{Schematic diagram of the hierarchical tracking control strategy for the fixed-wing UAV equipped with a turbojet engine.}
    \label{fig-1}
\end{figure*}

\subsection{Extended Disturbance Observer}
In this subsection, an extended disturbance observer for the system in the following form is presented, and will be utilized in the design of the tracking controller later.
\begin{equation}
    \label{equ-40}
    \dot{s} = f_s(s) s + g_s(s) u_s + w_{s0} d_s,
\end{equation}
where $s$ is the system state, $u_s$ is the system input, $d_s$ is the external disturbance, $w_{s0}$ is the constant disturbance coefficient, $f_s$ and $g_s$ are the system functions, all of which are appropriate dimensions.
The disturbance observer is inspired by the Luenberger observer \cite{Adamy2022} and designed as follows.
\begin{align}
    \label{equ-41}
    \dot{\hat{s}} &= f_s(s) s + g_s(s)u_s + w_{s0} \hat{d}_s + k_{o1} (s - \hat{s}), \notag \\
    \dot{\hat{d}}_s &= k_{o2} (s - \hat{s}),
\end{align}
where $\hat{s}$ and $\hat{d}$ are the estimations of the state and the disturbance, respectively. In \eqref{equ-41}, $k_{o1}$ and $k_{o2}$ are observer parameters. The observer matrix is defined as
\begin{equation}
    \label{equ-42}
    \boldsymbol{L}_o := \begin{bmatrix}
        -k_{o1} & w_{s0} \\
        -k_{o2} & 0
    \end{bmatrix}.
\end{equation}
Define the observer errors as $e_o := (e_s^T, e_d^T)^T$, where $e_s := s - \hat{s}$ and $e_d := d_s - \hat{d}_s$.
From \eqref{equ-40} and \eqref{equ-41}, we have
\begin{align}
    \label{equ-44}
    \dot{e}_s &= \dot{s} - \dot{\hat{s}} = -k_{o1} e_s + w_{s0} e_d, \notag \\
    \dot{e}_d &= \dot{d}_s - \dot{\hat{d}}_s = -k_{o2} e_s + h_s(t),
\end{align}
where $\dot{d}_s = h_s(t)$ and $\lim_{t \to \infty} h_s(t) = 0$ under Assumption~\ref{assum-1}. The following theorem shows the convergence of the observer error $e_o$.

\begin{theorem}
    \label{theorem-1}
    \rm
    Under Assumption~\ref{assum-1} and the disturbance observer designed as \eqref{equ-41}. If the observer matrix $\boldsymbol{L}_o$ is Hurwitz by selecting appropriate observer parameters $k_{o1}$ and $k_{o2}$, then the observer error $e_o$ is asymptotically stable at zero.
\end{theorem}

\begin{IEEEproof}
    
    The observer error $e_o$ is governed by \eqref{equ-44}, which can be rewritten into
    \begin{equation}
        \label{equ-45}
        \dot{e}_o = \boldsymbol{L}_o e_o + (0^T, h_s^T)^T.
    \end{equation}
    Under Assumption~\ref{assum-1}, we have $\lim_{t \to \infty} h_s(t) = 0$. Since $\boldsymbol{L}_o$ is Hurwitz, we can further follow \cite[Lemma~2]{Li2012} to derive that the observation error $e_o$ is asymptotically stable at zero.
\end{IEEEproof}

Theorem~1 shows that the observer error asymptotically converges to zero. Hence, the disturbance observer \eqref{equ-41} can be used in the controller design afterwards to implement disturbance dynamic compensation.

\subsection{Motion State Planner}

In this subsection, a motion state planner is proposed to calculate the desired motion state, including airspeed and flight direction, for tracking the desired position trajectories.

The position-tracking error is defined as $\boldsymbol{e}_p := \boldsymbol{p} - \boldsymbol{p}_d$, which is governed by
\begin{equation}
    \label{equ-46}
    \dot{\boldsymbol{e}}_p = \boldsymbol{v}_e + \boldsymbol{d}_p - \dot{\boldsymbol{p}}_d,
\end{equation}
where $\boldsymbol{v}_e = (v_x, v_y, v_z)^T$ is the velocity in the earth frame with $v_x, v_y, v_z \in \mathbb{R}$ as the velocity components. The motion state planner is designed as follows.
\begin{align}
    \label{equ-47}
    \boldsymbol{v}_{ed} &= (v_{xd}, v_{yd}, v_{zd})^T =  -k_p \boldsymbol{e}_p - \hat{\boldsymbol{d}_p} + \dot{\boldsymbol{p}}_d, \notag \\
    V_d &= \sqrt{v_{xd}^2 + v_{yd}^2 + v_{zd}^2}, \notag \\
    \gamma_d &= \arcsin(-v_{zd} / V_d), \notag \\
    \chi_d &= {\rm atan2}(v_{yd}, v_{xd}), 
\end{align}
where $k_p > 0$ is the controller gain, $\hat{\boldsymbol{d}_p}$ is the estimation of $\boldsymbol{d}_p$ obtained by the disturbance observer \eqref{equ-41}. In \eqref{equ-47}, ${\rm atan2}(\cdot)$ is the four-quadrant inverse tangent to ensure $\chi_d \in [-\pi, \pi]$.

The desired motion state is denoted as $\boldsymbol{v}_d = (V_d^T, \boldsymbol{\psi}^T_d)^T$, where $V_d$ is the desired airspeed and $\boldsymbol{\psi}_d = (\gamma_d, \chi_d)^T$ is the desired flight direction.
Define the motion state tracking error as $\boldsymbol{e}_v := \boldsymbol{v} - \boldsymbol{v}_{d}$, and the earth frame velocity tracking error as $\boldsymbol{e}_{ve} := \boldsymbol{v}_{e} - \boldsymbol{v}_{ed}$. We have the following lemma to establish the convergence of the position-tracking error $\boldsymbol{e}_p$.
\begin{lemma}
    \label{lem-2}
    \rm
    Consider the position subsystem \eqref{equ-31}. Under Assumption~\ref{assum-1}, if the motion state planner is designed as \eqref{equ-47} and the motion state tracking error $\boldsymbol{e}_v = 0$, then the position-tracking error $\boldsymbol{e}_p$ is asymptotically stable at zero.
\end{lemma}
\begin{IEEEproof}

    From \eqref{equ-46} and \eqref{equ-47}, we have
    \begin{equation}
        \label{equ-48}
        \dot{\boldsymbol{e}}_p = -k_p \boldsymbol{e}_p + \boldsymbol{e}_{po} + \boldsymbol{e}_{ve},
    \end{equation}
    where $\boldsymbol{e}_{po} := \boldsymbol{d}_p - \hat{\boldsymbol{d}_p}$ is the asymptotic observer error. From the definitions of the desired airspeed and flight direction in \eqref{equ-47}, we can further derive that $\boldsymbol{e}_{ve} = 0$ is equivalent to $\boldsymbol{e}_v = 0$.
    Consider the Lyapunov function candidate $V_1 = \boldsymbol{e}_p^T \boldsymbol{e}_p / 2$. If $\boldsymbol{e}_{v} = 0$ and $\boldsymbol{e}_{po} = 0$, by taking the derivative of $V_1$, we have
    \begin{equation}
        \label{equ-49}
        \dot{V}_1 = -k_p \boldsymbol{e}_p^T \boldsymbol{e}_p < 0, \ \text{if} \ \boldsymbol{e}_p \neq 0.
    \end{equation}
    From Theorem~\ref{theorem-1}, the observer error $\boldsymbol{e}_{po}$ is asymptotically stable at zero. Hence, we can further obtain that the position-tracking error $\boldsymbol{e}_p$ is asymptotically stable at zero if $\boldsymbol{e}_v = 0$ from Lemma~\ref{lem-1}.
\end{IEEEproof}

Lemma~\ref{lem-2} shows that the position-tracking error $\boldsymbol{e}_p$ asymptotically converges to zero under the disturbance compensation if the motion state tracking error $\boldsymbol{e}_{v} = 0$. Hence, what we need to do next is to control the motions such that the condition $\boldsymbol{e}_v = 0$ is guaranteed.

\subsection{Flight Direction Controller}
In this section, a flight direction controller is proposed, which aims to make the flight direction $\boldsymbol{\psi}$ converge to the desired value $\boldsymbol{\psi}_d = (\gamma_d, \chi_d)^T$ and the wind-axis bank angle $\mu$ converge to the desired value $\mu_d$.

{\bf Step 1:}
The flight direction tracking error is defined as $\boldsymbol{e}_\psi := \boldsymbol{\psi} - \boldsymbol{\psi}_d$. An auxiliary controller is designed as follows.
\begin{equation}
    \label{equ-50}
    \boldsymbol{F}_{ad} = -g_\psi^{-1} \big( k_\psi \boldsymbol{e}_\psi + f_\psi + \hat{\boldsymbol{d}_\psi} - \dot{\boldsymbol{\psi}}_d \big),
\end{equation}
where $k_\psi > 0$ is the controller gain, $\boldsymbol{F}_{ad}$ is the desired aerodynamic driving force, $\hat{\boldsymbol{d}_\psi}$ is the estimation of $\boldsymbol{d}_\psi$. From \eqref{equ-36}, we can further obtain the desired aerodynamic angle
\begin{equation}
    \label{equ-51}
    \boldsymbol{\theta}_d = f_a^{-1}(\boldsymbol{F}_{ad}).
\end{equation}
The aerodynamic force and angle tracking errors are defined as $\boldsymbol{e}_F := \boldsymbol{F}_{a} - \boldsymbol{F}_{ad}$ and $\boldsymbol{e}_\theta := \boldsymbol{\theta} - \boldsymbol{\theta}_d$, respectively. From \eqref{equ-35} and \eqref{equ-50}, we have
\begin{equation}
    \label{equ-52}
    \dot{\boldsymbol{e}}_\psi = -k_\psi \boldsymbol{e}_\psi + g_\psi \boldsymbol{e}_F + \boldsymbol{e}_{\psi o} = -k_\psi \boldsymbol{e}_\psi + g_\psi g_F \boldsymbol{e}_\theta + \boldsymbol{e}_{\psi o},
\end{equation}
where $\boldsymbol{e}_{\psi o} := \boldsymbol{d}_\psi - \hat{\boldsymbol{d}_\psi}$ is the observer error.

{\bf Step 2:}
In this step, the aerodynamic Euler angle $\boldsymbol{\Theta}$ aims to
track the desired value $\boldsymbol{\Theta}_d = (\boldsymbol{\theta}^T_d, \mu_d)^T$, which includes the auxiliary control law $\boldsymbol{\theta}_d$ and the desired bank angle $\mu_d$. To this end, we design the following auxiliary controller.
\begin{equation}
    \label{equ-53}
    \boldsymbol{\omega}_d = -g_\theta^{-1} \big(k_\theta (\boldsymbol{\Theta} - \boldsymbol{\Theta}_d) + f_\theta + \boldsymbol{\hat{d}_\Theta} - \dot{\boldsymbol{\Theta}}_d \big),
\end{equation}
where $k_\theta > 0$ is the controller gain, $\boldsymbol{\hat{d}_\Theta}$ is the estimation of $\boldsymbol{d}_\Theta$. Define the aerodynamic Euler angle tracking error as $\boldsymbol{e}_\Theta := \boldsymbol{\Theta} - \boldsymbol{\Theta}_d = (\boldsymbol{e}_\theta^T, e_\mu^T)^T$, and the angular velocity tracking error as $\boldsymbol{e}_\omega = \boldsymbol{\omega} - \boldsymbol{\omega}_d$. From \eqref{equ-35} and \eqref{equ-53}, we have
\begin{equation}
    \label{equ-54}
    \dot{\boldsymbol{e}}_\Theta = - k_\theta \boldsymbol{e}_\Theta + g_\theta \boldsymbol{e}_\omega + \boldsymbol{e}_{\theta o},
\end{equation}
where $\boldsymbol{e}_{\theta o} = \boldsymbol{d}_\Theta - \boldsymbol{\hat{d}_\Theta}$ is the observer error.

{\bf Step 3:}
Similar to Step 2, the angular velocity $\boldsymbol{\omega}$ aims to track the auxiliary control law $\boldsymbol{\omega}_d$ in this step. To this end, we design the controller as follows.
\begin{align}
    \label{equ-55}
    \boldsymbol{M}_d &= g_\omega^{-1} \big(k_\omega (\boldsymbol{\omega} - \boldsymbol{\omega}_d) + f_\omega + \hat{\boldsymbol{d}_\omega} - \dot{\boldsymbol{\omega}}_d \big), \notag \\
    \boldsymbol{\delta} &= g_M^{-1} (\boldsymbol{M}_d - f_M),
\end{align}
where $k_\omega > 0$ is the controller gain, $\boldsymbol{M}_d$ is the desired aerodynamic moment, and $\hat{\boldsymbol{d}_\omega}$ is the estimation of $\boldsymbol{d}_\omega$. From \eqref{equ-35} and \eqref{equ-55}, we have
\begin{equation}
    \label{equ-56}
    \dot{\boldsymbol{e}}_\omega = -k_\omega \boldsymbol{e}_\omega + \boldsymbol{e}_{\omega o},
\end{equation}
where $\boldsymbol{e}_{\omega o} := \boldsymbol{d}_\omega - \hat{\boldsymbol{d}_\omega}$ is the observer error.

The following lemma shows the convergence of the flight direction tracking error $\boldsymbol{e}_\psi$, the aerodynamic Euler angle tracking error $\boldsymbol{e}_\Theta$, and the angular velocity tracking error $\boldsymbol{e}_\omega$.

\begin{lemma}
    \label{lem-3}
    \rm
    Consider the flight direction subsystem \eqref{equ-35}. Under Assumption~\ref{assum-1} and the control law \eqref{equ-55} with auxiliary variables \eqref{equ-50}, \eqref{equ-51}, and \eqref{equ-53}, the tracking errors of flight direction $\boldsymbol{e}_\psi$, aerodynamic Euler angle $\boldsymbol{e}_\Theta$, and angular velocity $\boldsymbol{e}_\omega$ are all asymptotically stable at zero.
\end{lemma}
\begin{IEEEproof}

    First, consider the Lyapunov function candidate $V_\omega = \boldsymbol{e}_\omega^T \boldsymbol{e}_\omega / 2$. If $\boldsymbol{e}_\omega \neq 0$ and $\boldsymbol{e}_{\omega o} = 0$, the derivative of $V_\omega$ gives that
    \begin{equation}
        \label{equ-58}
        \dot{V}_\omega = - k_\omega \boldsymbol{e}_\omega^T \boldsymbol{e}_\omega < 0.
    \end{equation}
    From Theorem~\ref{theorem-1}, the observer error $\boldsymbol{e}_{\omega o}$ is asymptotically stable at zero. Hence, according to Lemma~\ref{lem-1}, we can further obtain that $\boldsymbol{e}_\omega$ is asymptotically stable at $\boldsymbol{e}_\omega = 0$.

    Next, consider the Lyapunov function candidate $V_\Theta = \boldsymbol{e}_\Theta^T \boldsymbol{e}_\Theta / 2$. If $\boldsymbol{e}_\Theta \neq 0$, $\boldsymbol{e}_\omega = 0$, and $\boldsymbol{e}_{\theta o} = 0$, we have the derivative of $V_\Theta$ as follows.
    \begin{equation}
        \label{equ-59}
        \dot{V}_\Theta = -k_\theta \boldsymbol{e}_\Theta^T \boldsymbol{e}_\Theta < 0.
    \end{equation}
    From Theorem~\ref{theorem-1} and Lemma~\ref{lem-1}, we can obtain that $(\boldsymbol{e}_\Theta^T, \boldsymbol{e}_\omega^T)^T$ is asymptotically stable at the origin.
    
    Since $\boldsymbol{e}_\theta$ is a part of $\boldsymbol{e}_\Theta$, $\boldsymbol{e}_\theta$ is also asymptotically stable at zero. Similarly, we further obtain that $(\boldsymbol{e}_\psi^T, \boldsymbol{e}_\Theta^T, \boldsymbol{e}_\omega^T)^T$ is asymptotically stable at the origin. That is, the tracking errors $\boldsymbol{e}_\psi$, $\boldsymbol{e}_\Theta$, and $\boldsymbol{e}_\omega$ are all asymptotically stable at zero.
\end{IEEEproof}

Lemma~\ref{lem-3} implies that the bank angle tracking error $e_\mu$ asymptotically converges to zero. Hence, both the flight direction $\boldsymbol{\psi}$ and the bank angle $\mu$ converge to their desired values $\boldsymbol{\psi}_d$ and $\mu_d$, respectively.

\subsection{Airspeed Controller}
An airspeed controller is proposed in this subsection to ensure that the airspeed $V$ converges to the desired value $V_d$, also known as the turbojet engine controller. In particular, we design an observer-based controller to compensate for the linearization error of the turbojet engine as follows.

{\bf Step 1:} In this step, a desired rotor speed generator is designed to control the airspeed $V$ track the desired value $V_d$. To this end, we first design the following auxiliary controller.
\begin{equation}
    \label{equ-60}
    \Delta T_d = -g_v^{-1} \big(k_v e_V + f_v + \hat{d}_{V} - \dot{V}_d \big),
\end{equation}
where $k_v > 0$ is the controller gain, $e_V := V - V_d$ is the airspeed tracking error, $\hat{d}_{V}$ is the estimation of $d_{V}$, and $\Delta T_d$ is the desired value of the relative thrust $\Delta T$. Define the thrust error as $e_{\Delta T}:= \Delta T - \Delta T_d$. We can further obtain that
\begin{equation}
    \label{equ-61}
    \dot{e}_V = -k_v e_V + g_v e_{\Delta T} + e_{ov},
\end{equation}
where $e_{ov} := d_V - \hat{d}_V$ is the observer error.
Next, the desired rotor speed $\Delta n_d$ is generated by the following system, in which we implement thrust compensation.
\begin{equation}
    \label{equ-62}
    \Delta \dot{n}_d = \frac{b_n}{d_n} \bigg(\frac{a_n d_n}{b_n} - c_n \bigg) \Delta n_d + \frac{b_n}{d_n} \Delta T_d + \hat{d}_{\Delta n} - \frac{b_n}{d_n} \hat{d}_{\Delta T},
\end{equation}
where $\hat{d}_{\Delta T}$ and $\hat{d}_{\Delta n}$ are estimations of $d_{\Delta T}$ and $d_{\Delta n}$, respectively.

{\bf Step 2:} In this step, an auxiliary controller is designed for the rotor speed $\Delta n$ as follows.
\begin{equation}
    \label{equ-63}
    \Delta W_f = -\frac{1}{b_n} \big( k_{\Delta n} (\Delta n - \Delta n_d) + a_n \Delta n - \Delta\dot{n}_d + \hat{d}_{\Delta n} \big),
\end{equation}
where $k_{\Delta n} > 0$ is the controller gain. From \eqref{equ-33} and \eqref{equ-63}, we have
\begin{equation}
    \label{equ-64}
    \dot{e}_{\Delta n} = - k_{\Delta n} e_{\Delta n} + e_{d\Delta n},
\end{equation}
where $e_{\Delta n}:= \Delta n - \Delta n_d$ is the rotation speed error, and $e_{d\Delta n} := d_{\Delta n} - \hat{d}_{\Delta n}$ is the observer error.

The following lemma shows the convergence of the airspeed tracking error $e_V$.

\begin{lemma}
    \label{lem-4}
    \rm
    Consider the airspeed subsystem \eqref{equ-33}. Under Assumption~\ref{assum-1} and the airspeed tracking control law \eqref{equ-63} with auxiliary variables \eqref{equ-60} and \eqref{equ-62}, the tracking error of the airspeed $e_V$ is asymptotically stable at zero.
\end{lemma}
\begin{IEEEproof}

    First, consider the Lyapunov function candidate $V_{\Delta n} = e_{\Delta n}^2 / 2$. If $e_{d\Delta n} = 0$, by taking the derivative of $V_{\Delta n}$, we have
    \begin{equation}
        \label{equ-65}
        \dot{V}_{\Delta n} = -k_{\Delta n} e_{\Delta n}^2 < 0.
    \end{equation}
    From Theorem~\ref{theorem-1} and Lemma~\ref{lem-1}, the rotor speed error $e_{\Delta n}$ is asymptotically stable at zero.
    Next, from \eqref{equ-62}, we have
    \begin{equation}
        \label{equ-66}
        \Delta T_d = \bigg(c_n - \frac{a_n d_n}{b_n} \bigg) \Delta n_d + \frac{d_n}{b_n} \Delta \dot{n}_d - \frac{d_n}{b_n} \hat{d}_{\Delta n} + \hat{d}_{\Delta T}.
    \end{equation}
    Combining \eqref{equ-33} and \eqref{equ-66}, we can further derive that
    \begin{align}
        \label{equ-67}
        e_{\Delta T} &= \Delta T - \Delta T_d \notag \\
        &= \bigg(c_n - \frac{a_n d_n}{b_n} - \frac{k_{\Delta n} d_n}{b_n} \bigg) e_{\Delta n} + e_{d\Delta T},
    \end{align}
    where $e_{d\Delta T}:= d_{\Delta T} - \hat{d}_{\Delta T}$ is the observer error.
    Since $e_{\Delta n}$ is asymptotically stable at zero and $e_{d\Delta T}$ converges to zero by Theorem 1, $e_{\Delta T}$ is also asymptotically stable at zero. Similarly, we further obtain that $e_V$ in \eqref{equ-61} is asymptotically stable at zero by Theorem~\ref{theorem-1} and Lemma~\ref{lem-1}.
\end{IEEEproof}

\subsection{Closed-loop Stability Analysis}
From the above subsections, the modules of the proposed integrated flight and propulsion control framework for fixed-wing UAVs are fully designed. The following theorem discusses the closed-loop stability.

\begin{theorem}
    \label{theorem-2}
    \rm
    Consider the fixed-wing UAV and turbojet engine modeled as \eqref{equ-31}, \eqref{equ-33}, and \eqref{equ-35}. Suppose Assumption~\ref{assum-1} hold. If the motion state planner is designed as \eqref{equ-47}, the flight direction control module is designed as \eqref{equ-55}, the airspeed control module is designed as \eqref{equ-63}, and the disturbance observer is designed as \eqref{equ-41}, then the control objective in \eqref{equ-39} is achieved, i.e., $\lim_{t \to \infty} \boldsymbol{e}_p = 0$ and $\lim_{t \to \infty} e_\mu = 0$.
\end{theorem}
\begin{IEEEproof}

    From Lemmas~\ref{lem-3} and \ref{lem-4}, the flight direction tracking error $\boldsymbol{e}_\psi$, the aerodynamic Euler angle tracking error $\boldsymbol{e}_\Theta$, and the airspeed tracking error $e_V$ are asymptotically stable at zero. Hence, the bank angle tracking error $e_\mu$ is asymptotically stable at zero, and the velocity tracking error $\boldsymbol{e}_v = (e_V^T, e_\theta^T)^T$ is asymptotically stable at $\boldsymbol{e}_v = 0$.
    From Lemmas~\ref{lem-1} and \ref{lem-2}, we can further obtain that the position-tracking error $\boldsymbol{e}_p$ is asymptotically stable at $\boldsymbol{e}_p = 0$. That is, $\lim_{t \to \infty} \boldsymbol{e}_p = 0$ and $\lim_{t \to \infty} e_\mu = 0$.
\end{IEEEproof}

Theorem~\ref{theorem-2} shows that the position and the bank angle of the aircraft converge to their desired trajectories, respectively. In practice, the desired bank angle is designed to meet the conditions for a coordinated turn strategy, thereby enhancing flexibility and robustness \cite{Song2021}. The coordinated turn strategy adopted in this paper is designed as \cite{Stevens2015}
\begin{equation}
    \label{equ-68}
    \mu_d = \arctan\bigg( \frac{V \dot{\chi}_d}{g} \bigg).
\end{equation}
The coordinated turn module is also shown in Fig.~\ref{fig-1} as a part of the hierarchical tracking control framework.

\section{SIMULATION RESULTS}
\label{sec-4}

\begin{figure*} [!t]
    \centering
    \includegraphics[width=16.5cm]{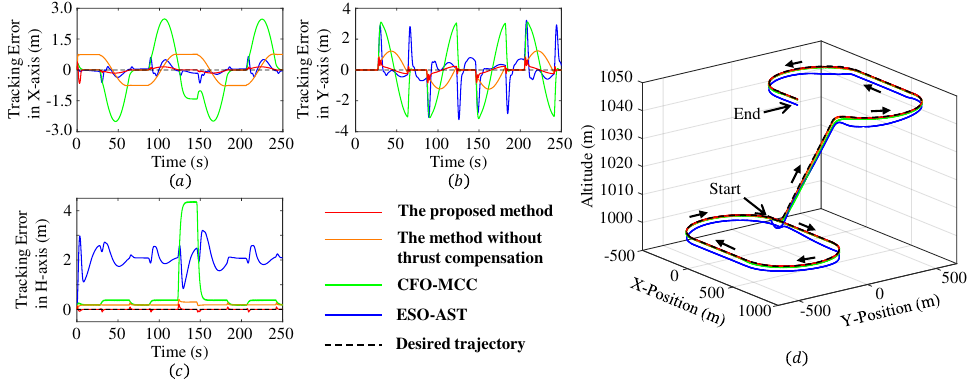}
    \caption{The illustration of the position-tracking errors and trajectories under different control strategies. 
    }
    \label{fig-3}
\end{figure*}

\begin{figure*} [!t]
    \centering
    \includegraphics[width=17.5cm]{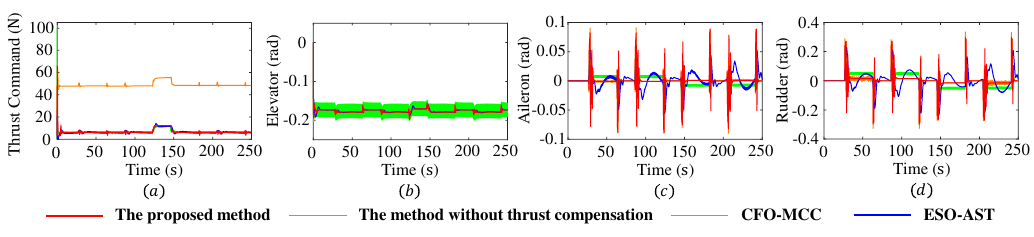}
    \caption{The illustration of the fixed-wing UAV control inputs. 
    }
    \label{fig-4}
\end{figure*}

\begin{figure} [!t]
    \centering
    \includegraphics[width=6.5cm]{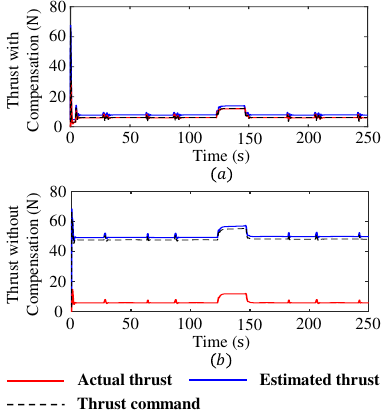}
    \caption{The illustration of the thrust in both with and without thrust compensation cases. 
    }
    \label{fig-5}
\end{figure}

In this section, the proposed control strategy is verified through a numerical simulation. This simulation illustrates the advantages of the proposed controller in addressing the time-varying position-tracking problem and unmodeled thrust dynamics. The proposed control strategy is compared with the adaptive super twisting controller with extended state observer (ESO-AST) in \cite{Castañeda2017}, and the model compensation controller with compensation function observer (CFO-MCC) in \cite{Xu2024}.

In this simulation, the desired trajectory is designed as the black dotted line in Fig.~\ref{fig-3}(d), which includes maneuvers such as climbing and turning. In particular, the aircraft first moves clockwise, then climbs to a certain altitude, and finally moves counterclockwise.

Under different controllers, the simulation results are shown in Figs.~\ref{fig-3}-\ref{fig-5}. Fig.~\ref{fig-3} shows the position-tracking errors for the desired trajectory. As shown in Fig.~\ref{fig-3}, the proposed method enables the aircraft to accurately track the desired trajectory, with the tracking errors converging to zero, even during maneuvers. The control inputs are presented in Fig.~\ref{fig-4}. Fig.~\ref{fig-5} illustrates the thrust generated by the turbojet engine during the simulation. From Fig.~\ref{fig-5}(a), we can see that with the proposed thrust compensation, the estimated thrust converges to the actual thrust, thereby ensuring the convergence of the actual thrust to the thrust command.

If the existing ESO-AST or CFO-MCC is applied, the tracking error cannot converge to zero, and the desired performance cannot be guaranteed, as the blue and green curves in Fig.~\ref{fig-3}. In particular, the ESO-AST brings steady-state error in the altitude channel because only attitude and airspeed are considered, rather than flight velocity. Furthermore, with the CFO-MCC method, periodic tracking errors appear in the x-axis and y-axis, and an increase in error is observed in the H-axis over the time interval $[125, 150]$. These issues arise because the method is designed exclusively for a constant desired value. In these two cases, the control inputs are not convergent either, especially those of the CFO-MCC method, which exhibit high-frequency oscillations and are therefore impractical for real-world applications, as shown in Fig.~\ref{fig-4}. From these two cases, we can see the advantages of the proposed control method over the existing ones in dealing with the time-varying position-tracking problem.

If the thrust compensation is not involved in the controller design, the tracking errors cannot converge to zero, particularly in the x-axis and y-axis, as the orange curve shown in Fig. \ref{fig-3}. In Fig.~\ref{fig-4}(a), the thrust command for the method without thrust compensation is distinctly different from the commands required by the other methods. This is because the estimated thrust cannot converge to the actual thrust as shown in Fig.~\ref{fig-5}(b). In particular, inaccurate thrust estimation leads to an increase in the thrust command, resulting in steady-state errors. From this case, we can see the necessity of thrust compensation and the advantages of the proposed method in dealing with the unmeasurable thrust dynamics.

\section{CONCLUSION}
\label{sec-5}

In this paper, a hierarchical tracking control framework with thrust and disturbance compensation is proposed for fixed-wing UAVs equipped with turbojet engines. To begin with, we propose a perturbed fixed-wing UAV model that consists of fixed-wing UAV kinetics and turbojet engine dynamics, and we consider both the unmodeled dynamics and external disturbances. Next, a hierarchical control framework is proposed to guarantee the desired tracking performance and to deal with the unmodeled dynamics and disturbances. More specifically, an extended observer is designed to estimate both thrust and disturbances within a unified framework, and then three observer-based controllers are designed. Finally, the effectiveness of the proposed control strategy is illustrated by a comparative numerical simulation. We will focus on extending the proposed method to the field of maneuvering multiple vehicles in the near future.

\bibliographystyle{IEEEtran}
\bibliography{reference}

\end{document}